\documentclass{article}

\usepackage[utf8]{inputenc} 
\usepackage[T1]{fontenc}    
\usepackage{url}            
\usepackage{booktabs}       
\usepackage{amsfonts}       
\usepackage{nicefrac}       
\usepackage{microtype}      
\usepackage{lipsum}		
\usepackage{graphicx}
\usepackage{natbib}
\usepackage{doi}

\usepackage{amsthm,amsmath,amssymb}
\usepackage{enumerate}
\usepackage{etoolbox}
\usepackage{mathtools}
\usepackage[table]{xcolor}
\usepackage{tcolorbox}
\usepackage{tikz-cd}
\usepackage{changepage}
\usepackage{algpseudocodex}

\usepackage[nameinlink]{cleveref}
\creflabelformat{equation}{#2\textup{#1}#3}

\newcommand{\Z}{\mathbb{Z}}

\newcommand{\C}{\mathbb{C}}

\newcommand{\Q}{\mathbb{Q}}
\newcommand{\Qbar}{\overline{\Q}}
\newcommand{\R}{\mathbb{R}}

\newcommand{\I}{\mathbb{I}}
\newcommand{\V}{\mathbb{V}}

\newcommand{\innerproduct}[2]{\langle #1, #2 \rangle}
\renewcommand{\bold}{\boldsymbol}

\DeclareMathOperator{\Rank}{Rank}

\DeclareMathOperator{\Span}{Span}

\newtheorem{theoremcount}{Theoremcount}[section]
\newtheorem{conjcount}{Conjcount}

\theoremstyle{plain}
\newtheorem{theorem}[theoremcount]{Theorem}
\newtheorem*{theorem*}{Theorem}
\newtheorem{lemma}[theoremcount]{Lemma}
\newtheorem{proposition}[theoremcount]{Proposition}
\newtheorem{corollary}[theoremcount]{Corollary}
\newtheorem{conjecture}[conjcount]{Conjecture}

\theoremstyle{definition}

\newtheorem{example}[theoremcount]{Example}

\crefname{theoremcount}{Theorem}{Theorems}
\crefname{conjcount}{Conjecture}{Conjectures}
\crefname{conj1count}{Conjecture}{Conjectures}
\crefname{quescount}{Question}{Questions}
\crefname{theorem}{Theorem}{Theorems}
\crefname{lemma}{Lemma}{Lemmas}
\crefname{proposition}{Proposition}{Propositions}
\crefname{corollary}{Corollary}{Corollaries}
\crefname{definition}{Definition}{Definitions}
\crefname{example}{Example}{Examples}
\crefname{remark}{Remark}{Remarks}
\crefname{conjecture}{Conjecture}{Conjectures}
\crefname{question}{Question}{Questions}

\date{March 7, 2026}

\author{ \href{https://orcid.org/0009-0001-0062-648X}{\includegraphics[scale=0.06]{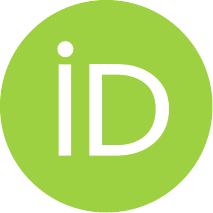}}\hspace{1mm}Igor V.~Loo \thanks{School of Physical and Mathematical Sciences, Nanyang Technological University, Singapore, email:
\texttt{IGOR002@e.ntu.edu.sg}.
The work of Igor Van Loo was supported by an NTU Provost Graduate Award, whose support is gratefully acknowledged.} \hspace{2mm} and 
\href{https://orcid.org/0000-0003-3141-3118}{\includegraphics[scale=0.06]{orcid.pdf}}\hspace{1mm}Fr\'ed\'erique Oggier
}

\hypersetup{
pdftitle={On the Existence of Algebraic Equiangular Lines},
pdfsubject={q-bio.NC, q-bio.QM},
pdfauthor={Igor V.~Loo, Frederique E.~Oggier},
pdfkeywords={SIC-POVMs; Equiangular Lines; Hilbert's Nullstellensatz; Real Algebraic Geometry},
}

\begin{document}
\title{On the Existence of Algebraic Equiangular Lines}

\maketitle

\begin{abstract}
We consider real and complex equiangular lines, generated by unit vectors. We show that, for an arbitrary dimension $d$, if there exists a set of $d^2$ equiangular unit vectors in $\C^d$, then there must exist a set of $d^2$ equiangular unit vectors with all of their coefficients in a number field. This result is motivated by the question of constructing SIC-POVMs in quantum physics and conjectures around them. We discuss applications of our techniques to the case of real equiangular lines and consequences of the above results.
\end{abstract}


%
%
\section{Introduction}

A set of complex or real lines passing through the origin is called equiangular if the angle between any two lines is constant. The case of a complex or real orthonormal basis shows the existence of $d$ equiangular lines in $\C^d$, respectively in $\R^d$. We may thus ask about the maximal number of equiangular lines in a given dimension $d$.

For the complex case, it is known \cite{Godsil_2009_equiangular_proof} that the size of a set of equiangular lines in $\C^d$ cannot exceed $d^2$.
In the quantum physics literature, a maximal set of $d^2$ equiangular lines in $\C^d$ is the underlying
mathematical object defining a symmetric informationally complete positive-operator-valued measure (SIC-POVM). It is conjectured that the bound of $d^2$ equiangular lines is achievable for every dimension $d\geq 2$, as first proposed in 1999 by Gerhard Zauner in his PhD thesis \cite{zauner_2011}, and since then numerous works have provided numerical and exact constructions further strengthening the conjecture (see e.g., \cite{Renes_2004} \cite{scott_grassl_2009} \cite{scott2017sicsextendinglistsolutions} \cite{Grassl_2017_FibLucas} \cite{Appleby_2022_StarkUnit} \cite{bengtsson2024sicpovmsstarkunitsdimensions}). We refer the reader to \cite{Fuchs_2017} for an excellent overview on this topic and to \cite{appleby2025constructiveapproachzaunersconjecture} for the current research frontier.

Many of these constructions have rather mysterious connections with algebraic number theory (see e.g., \cite{Appleby_2017} \cite{appleby2012galoisautomorphismssymmetricmeasurement} \cite{Bengtsson_2017}) and geometry, leading to a number of conjectures, such as (see \Cref{eq:normalizedoverlap} for the definition of normalized overlap, and \Cref{eq:wh_fid_cond} for that of covariance under the action of the Weyl-Heisenberg group):

\begin{conjecture}\cite[p.~7]{appleby2025constructiveapproachzaunersconjecture}\label{conj:overlap}
The normalized overlaps $ e^{i\theta_{jl}}$ are algebraic units.
\end{conjecture}

\begin{conjecture}\cite[p.~5]{appleby2025constructiveapproachzaunersconjecture} \label{conj:wh}
Numerical evidence suggests that the polynomial equations defining a SIC-POVM which is covariant under the action of the Weyl-Heisenberg group define an algebraic variety of dimension zero for $d > 3$.    
\end{conjecture}

Our main result in the complex case proves, loosely speaking, that the  existence of $d^2$ equiangular lines in $\C^d$ implies the existence of $d^2$ equiangular algebraic lines (see \Cref{thm:existance_of_real_alg_lines_in_Cd} for the precise statement), providing a formal motivation for the study of complex equiangular lines using algebraic number theory for all dimensions.

For the real case, it is known \cite[Theorem 3.5]{LemmensSeidel} that the size of a set of equiangular lines in $\R^d$ cannot exceed $\tfrac{d(d+1)}{2}$, though unlike in the complex case, it has been shown that the bound cannot always be attained (for example for $d = 4$, the maximal set of real equiangular lines has size $6$ opposed to $\frac{4\cdot 5}{2} = 10$). We will illustrate how the methodology developed for the complex case applies to the real case, in the process highlighting differences between the complex and real equiangular lines.

We start by precisely defining the problem of constructing equiangular lines in \Cref{sec:polynomialequations}, where we reformulate it as solving systems of polynomial equations. This formulation allows us to invoke results from real algebraic geometry in \Cref{sec:algebraic} leading to
a proof that the existence of complex equiangular lines implies the existence of algebraic equiangular lines, while we make explicit results from complex algebraic geometry that imply that an algebraic variety of dimension zero implies algebraic solutions.
Consequences of these results, in particular pertaining to \Cref{conj:overlap} and \Cref{conj:wh}, are discussed in \Cref{sec:consequences}.

%
%
%
\section{Reformulating Angle Constraints as a System of Polynomial Equations}
\label{sec:polynomialequations}

In the complex case, for two vectors $\bold{u} = (u_1, u_2, \cdots, u_d)^T$ and $\bold{v} = (v_1, v_2, \cdots ,v_d)^T$ in $\C^d$, the angle between them is by definition
 \begin{equation}  \label{eq:2.4}
        \arccos\left(\frac{|\innerproduct{\bold{u}}{\bold{v}}|}{||\bold{u}|| ||\bold{v}||}\right)
    \end{equation}
where 
 $\innerproduct{\bold{u}}{\bold{v}} = \sum_{j = 1}^d u_j\overline{v_j}$, $\overline{v_j}$ denotes the complex conjugate of $v_j$ and 
$||\bold{u}|| = \sqrt{\innerproduct{\bold{u}}{\bold{u}}}$.

A set of lines $\{l_1, l_2, \cdots , l_n\}$ which pass through the origin of $\C^d$ is called equiangular if the angle between any two lines is constant. Every line $l_j$ is spanned by a unit vector $\bold{u}_j$, therefore the angle between two lines $l_j,l_l$ is simplified to be $\arccos(|\innerproduct{\bold{u}_j}{\bold{u}_l}|) = c$ by \Cref{eq:2.4}. If we let $|\innerproduct{\bold{u}_j}{\bold{u}_l}| = \cos(c) = \alpha < 1$, an equivalent definition for a set of equiangular lines is a set $\{\bold{u}_1,\ldots,\bold{u}_n\}$ of unit vectors in $\C^d$ satisfying
    \begin{equation*} 
        |\innerproduct{\bold{u}_j}{ \bold{u}_l}| =   
    \begin{cases}
        1, \quad j = l \\
        \alpha, \quad j \neq l.
    \end{cases}
    \end{equation*}

It is known that the size of a set of equiangular lines cannot exceed $d^2$ in $\C^d$ \cite{Godsil_2009_equiangular_proof} and that a set $\{\bold{u}_1, \bold{u}_2, \cdots, \bold{u}_{d^2}\}$ of unit vectors in $\C^d$ represents a set of equiangular lines if and only if 
\begin{equation}\label{eq:condition_for_maximal_set_of_equiangular_lines}
       |\innerproduct{\bold{u}_j}{\bold{u}_l}|^2 = \begin{cases}
           1, \quad j = l \\
           \frac{1}{d + 1} \quad j \neq l.
       \end{cases}
    \end{equation}
When there is no confusion and by abuse of language, we may refer to such a maximal set of complex equiangular lines as a SIC-POVM. Expressions such as algebraic SIC-POVM or algebraic equiangular lines are also abuse of language, and short expressions to say that the underlying unit vectors generating the equiangular lines have coefficients all of which are algebraic.

The real case is similarly defined as in the complex case, that is a set of lines $ \{l_1, l_2, \cdots , l_n\}$ which pass through the origin of $\R^d$ is called equiangular if the angle between any two lines is constant. Every line $l_j$ is spanned by a unit vector $\bold{u}_j$, and a set of real lines is equiangular if there exists an $\alpha \in \R$ such that 
\begin{equation}
\label{eq:real_equiangular_line_condition}
    |\innerproduct{\bold{u}_j}{ \bold{u}_l}| =   
    \begin{cases}
        1, \quad j = l \\
        \alpha, \quad j \neq l
    \end{cases}
\end{equation}
where the inner product is the usual Euclidean inner product $\innerproduct{x}{y} = x^Ty$.


\subsection{Basic polynomial equations for the complex case}

We rewrite the necessary and sufficient condition defined by \Cref{eq:condition_for_maximal_set_of_equiangular_lines} as a polynomial system.

Set $\bold{u}_j = (u_{j1}, \ldots, u_{jd})^T$ for $j = 1, \ldots, d^2$ and let $A_{j, k} = \Re(u_{j, k}),~B_{j, k} = \Im(u_{j, k})$ be respectively the real and imaginary parts of $u_{jk}$,  for $k = 1, \ldots, d$, that is
\begin{equation*}
    \bold{u}_j = \begin{pmatrix}
        u_{j1} \\
        \vdots \\
        u_{jd}
    \end{pmatrix} = \begin{pmatrix}
        A_{j ,1} \\
        \vdots \\
        A_{j ,d}
    \end{pmatrix} + i \begin{pmatrix}
        B_{j ,1} \\
        \vdots \\
        B_{j ,d}
    \end{pmatrix}.
\end{equation*}

We rewrite the inner product as
\begin{align*}
\begin{split}
    \innerproduct{\bold{u}_j}{\bold{u}_l} 
    &= \sum_{k = 1}^d (A_{j,k}+iB_{j,k})(A_{l,k}-iB_{l,k}) \\
    &= \sum_{k = 1}^d (A_{j ,k}A_{l ,k} + B_{j, k}B_{l, k}) + i \sum_{k = 1}^d (B_{j ,k}A_{l ,k} - A_{j ,k}B_{l, k}) \\ 
    &= \bold{C}_{j, l} + i\bold{D}_{j, l}.
\end{split}
\end{align*}

This implies that
\begin{equation*}
    |\innerproduct{\bold{u}_j}{\bold{u}_l}|^2 = \bold{C}_{j, l}^2 + \bold{D}_{j, l}^2 = p_{j, l} \in \Q[A_{1, 1}, \ldots , B_{d^2, d}]
\end{equation*} 
is a polynomial in the variables $A_{1, 1}, \ldots, B_{d^2, d}$ with coefficients in the field $\Q$.
Thus,  a set $\{\bold{u}_1, \bold{u}_2, \cdots, \bold{u}_{d^2}\}$ of unit vectors in $\C^d$ representing a maximal set of equiangular lines corresponds to a real solution to the polynomial equation system
\begin{equation}
\label{eq:psic}
      p_{jl}(A_{1,1},\ldots,B_{d^2,d}) = \begin{cases}
           1, \quad j = l \\
           \frac{1}{d + 1} \quad j \neq l.
       \end{cases}
    \end{equation}

More precisely, the existence of $d^2$ unit vectors in $\C^d$ representing a set of equiangular lines means that we can find solutions in $\R$ for the variables $A_{1,1},\ldots,B_{d^2,d}$ (which are real parts of complex coefficients).


\subsection{Polynomial equations for Weyl-Heisenberg covariant SIC-POVMs}
\label{subsec:poly}

Let $\Z_d$ be the integers modulo $d$ and $I_d$ be the $d$-dimensional identity matrix. 
Set $\omega_d = e^{2\pi i/d}$ and define the $d$-dimensional unitary matrices

\begin{equation*}
    \bold{U} = 
    \begin{pmatrix}
        1 & 0 & 0 & \cdots & 0 \\
        0 & \omega_d & 0 & \cdots & 0 \\
        0 & 0 & \omega_d^2 & \cdots & 0 \\
        \vdots & \vdots & \vdots & \ddots & \vdots \\
        0 & 0 & 0 & \cdots & \omega_d^{d - 1} \\
    \end{pmatrix}, \quad
    \bold{V} = 
    \begin{pmatrix}
        0 & 1 & 0 & 0 &\cdots & 0 \\
        0 & 0 & 1 & 0 & \cdots & 0 \\
        \vdots & \vdots & \vdots & \vdots & \ddots & \vdots \\
        0 & 0 & 0 & 0 & \cdots & 1 \\
        1 & 0 & 0 & 0 & \cdots & 0
    \end{pmatrix}.
\end{equation*}
The following relations hold:
\begin{align*}
    \begin{split}
    \bold{U}^d &= \bold{V}^d = I_d
    \end{split} \\
    \begin{split}
    \bold{V}^a\bold{U}^b &= \omega_d^{ab}\bold{U}^b\bold{V}^a, \quad \forall a, b \in \Z_d.
    \end{split}
\end{align*}
The matrices $\bold{V}^a\bold{U}^b$, $a,b\in\Z_d$, are known as Weyl matrices, and the group generated by $\bold{U}$ and $\bold{V}$ is called the Weyl-Heisenberg group 
denoted by $H(d)$ and is explicitly given by
\begin{equation*}
    H(d) = \langle \bold{U}, \bold{V} \rangle = \{\omega_d^c\bold{V}^a\bold{U}^b : a, b, c \in \Z_d\}.
\end{equation*}
Its center is $Z(H(d)) = \{\omega_d^iI_d : 0 \leq i \leq d-1\}$
and the quotient $H(d)/Z(H(d)) $ is a group of order $d^2$ given by

\begin{equation*}
    H(d)/Z(H(d)) = \{\bold{V}^a\bold{U}^b : a, b \in \Z_d\}.
\end{equation*}

A fiducial vector $\bold{v} = (v_0, \ldots, v_{d - 1})^T$ is a vector such that the $d^2$ vectors $\bold{v}_{a, b} = \bold{V}^a\bold{U}^b\bold{v}$, $a,b\in\Z_d$ represent a set of $d^2$ equiangular lines. The resulting SIC-POVM is often referred to as being a Weyl-Heisenberg covariant SIC-POVM, or as being covariant under the action of the Weyl-Heisenberg group, and the corresponding fiducial vector a Weyl-Heisenberg fiducial vector. 
Since $H(d)/Z(H(d))$ forms a group, the necessary and sufficient conditions defined by \Cref{eq:condition_for_maximal_set_of_equiangular_lines} become the following:
\begin{equation}
\label{eq:wh_fid_cond}
    |\innerproduct{\bold{v}_{a, b}}{\bold{v}}|^2 = \begin{cases}
        1, \quad a = b = 0 \\
        \frac{1}{d + 1}, \quad (a, b) \in \Z_d^2/\{(0, 0)\}.
    \end{cases}
\end{equation}

We decompose the real and imaginary part of each coefficient, using $0$-indexing to simplify modular arithmetic, as follows:
\begin{equation*}
    \bold{v} = \begin{pmatrix}
        v_0 \\
        \vdots \\
        v_{d - 1}
    \end{pmatrix} = \begin{pmatrix}
        A_0 \\
        \vdots \\
        A_{d - 1}
    \end{pmatrix} + i\begin{pmatrix}
        B_0 \\
        \vdots \\
        B_{d - 1}
    \end{pmatrix}.
\end{equation*}

Now using 
\begin{align*}
    (\bold{V}^a \bold{u})_k &= u_{a + k ({\rm mod}~d)},  \\
    (\bold{U}^b \bold{v})_k &= \omega_d^{bk}v_k,
\end{align*}
we have that
\begin{equation*}
(\bold{v}_{a,b})_k = 
    (\bold{V}^a \bold{U}^b \bold{v})_k = \omega_d^{b(a + k)}v_{a + k ({\rm mod}~d) }.
\end{equation*}
   
Substituting $j = k + a$ and considering $j-a \pmod{d}$, we have
\begin{align*}
\begin{split}
    \innerproduct{\bold{v}_{a, b}}{\bold{v}} &=\sum_{j = 0}^{d - 1} \omega_d^{bj}v_j \overline{v}_{j - a}\\   
    & =\sum_{j = 0}^{d - 1} \omega_d^{bj} 
    [(A_jA_{j - a} + B_jB_{j - a}) + i(A_{j - a}B_j - A_jB_{j - a}) ] \\
    & = \sum_{j = 0}^{d - 1} \omega_d^{bj} (C_j^{(a)} + i D_j^{(a)}).
\end{split}
\end{align*}
where $C_j^{(a)} = A_jA_{j - a} + B_jB_{j - a}$ and $D_j^{(a)} = A_{j - a}B_j - A_jB_{j - a}$.
       
Note that $\omega_d^{bj} = \Re(\omega_d^{bj})+i\Im(\omega_d^{bj})$, 
$ \omega_d^{bj} + \overline{\omega_d^{bj}}=2\cos(\frac{2bj\pi}{d})$, and
$ \omega_d^{bj} - \overline{\omega_d^{bj}}=2i\sin(\frac{2bj\pi}{d})$
so
\begin{equation*}
    \omega_d^{bj} = \frac{\omega_d^{bj} + \overline{\omega_d^{bj}}}{2} + i \frac{\omega_d^{bj} - \overline{\omega_d^{bj}}}{2i} = \alpha_{bj} + i \beta_{bj}.
\end{equation*}
 
We thus get
\begin{align*}
\begin{split}
    \innerproduct{\bold{v}_{a, b}}{\bold{v}} &= \sum_{j = 0}^{d - 1}(\alpha_{bj} + i \beta_{bj})(C_j^{(a)} + iD_j^{(a)}) \\
    &= \sum_{j = 0}^{d - 1}C_j^{(a)} \alpha_{bj} - D_j^{(a)}\beta_{bj} + i(C_j^{(a)}\beta_{bj} + D_j^{(a)} \alpha_{bj})\\ 
    &= \sum_{j = 0}^{d - 1}(C_j^{(a)} \alpha_{bj} - D_j^{(a)}\beta_{bj}) + i\sum_{j = 0}^{d - 1}(C_j^{(a)}\beta_{bj} + D_j^{(a)} \alpha_{bj}) \\
    &= \bold{C}_{a, b} + i \bold{D}_{a, b}
\end{split}
\end{align*}
where $\bold{C}_{a, b}, \bold{D}_{a, b}$ are polynomials in $A_0,\ldots,A_{d-1}, B_0,\ldots,B_{d-1}$ whose coefficients are $\alpha_{bj},\beta_{bj}$, living, respectively, in $\Q(\omega_d+\overline{\omega_d})=\Q(\cos(\tfrac{2\pi}{d}))$ and $\Q(\sin(\tfrac{2\pi}{d}))$.
It follows that
\begin{align*}
\begin{split}
    |\innerproduct{\bold{v}_{a, b}}{\bold{v}}|^2 &= \bold{C}_{a, b}^2 + \bold{D}_{a, b}^2 \\
    &= p_{a, b} \in \Q(\cos(\tfrac{2\pi}{d}), \sin(\tfrac{2\pi}{d})) [A_0, \ldots, A_{d - 1}, B_0, \ldots, B_{d - 1}].
\end{split}
\end{align*}

Thus, a vector $\bold{v} = (v_0, \ldots, v_{d - 1})^T$ representing a fiducial vector means a set of real solutions to the polynomial equation system
\begin{equation}
\label{eq:fid}
   p_{a,b}(A_0,\ldots,A_{d-1},B_0,\ldots,B_{d-1}) = \begin{cases}
        1, \quad a = b = 0 \\
        \frac{1}{d + 1}, \quad (a, b) \in \Z_d^2/\{(0, 0)\}.
    \end{cases}
\end{equation}


\subsection{Polynomial equations for the real case}

As in the complex case, we construct the polynomial equations for the system using \Cref{eq:real_equiangular_line_condition}. Set
\begin{equation*}
    \bold{u}_j = \begin{pmatrix}
        u_{j, 1} \\
        \vdots \\
        u_{j, d}
    \end{pmatrix}.
\end{equation*}
For $1 \leq j \leq N$ where $N$ is the number of lines we have in $\R^d$. Then we write the inner product as
\begin{equation*}
    \innerproduct{\bold{u}_j}{\bold{u}_l} = \sum_{k = 1}^d u_{j, k}u_{l, k} = \bold{C}_{j, l}
\end{equation*}
hence we have
\begin{equation*}
    |\innerproduct{\bold{u}_j}{\bold{u}_l}|^2 = \bold{C}_{j, l}^2 = p_{j, l} \in \Q[u_{1, 1}, \ldots, u_{N, d}]
\end{equation*}
is a polynomial in the variables $u_{1, 1}, \ldots, u_{N, d}$ with coefficients in the field $\Q$. Thus, a set $\{\bold{u}_1, \bold{u}_2, \cdots, \bold{u}_{N}\}$ of unit vectors in $\R^d$ (with an angle $\alpha$ between them) representing a set of equiangular lines means a set of solutions to the polynomial equation system
\begin{equation}\label{eq:realsystem}
  p_{j,l}(u_{1,1},\ldots,u_{N,d}) = \begin{cases}
       1, \quad j = l \\
       \alpha \quad j \neq l.
   \end{cases}
\end{equation}

It is known \cite[Theorem 3.4]{LemmensSeidel} that if $N>2d$, then  
$1/\alpha$ is an odd integer.

%
%
%

\section{Existence of Algebraic Solutions}
\label{sec:algebraic}

It is conjectured that $d^2$ equiangular lines always exist in $\C^d$ \cite{zauner_2011}.  Assuming this conjecture is true, there exist real solutions to the polynomial systems defined by \Cref{eq:psic} and \Cref{eq:fid}. 
Our goal is to show that under this assumption, there exist solutions that live in $\R \cap \overline{\Q}$, where 
$\overline{\Q}$ is the algebraic closure of $\Q$. By definition, it is an algebraic extension of $\Q$ that is algebraically closed (meaning that every non-constant polynomial with coefficients in this field has a root in it). We recall that being algebraic means that every element in $\overline{\Q}$ is the root of some polynomial with coefficients in $\Q$. The field $\overline{\Q}$ is also called the field of algebraic numbers. We may think of $\overline{\Q}$ as the union of all number fields, that is $\overline{\Q} = \bigcup K$, where $K$ runs over all possible finite extensions of $\Q$.

We also recall that the definition of a real closed field is being a real field (meaning that $-1$ is not a sum of squares) that has no non-trivial real algebraic extension \cite[Definition 1.2.1]{Bochnak_Coste_Roy_1998}. An important example for our case is the field $F = \R \cap \overline{\Q}$ is a real closed field, which can be proved to be real closed by noting that $F[i] = \Qbar$ \cite[Theorem 1.2.2]{Bochnak_Coste_Roy_1998}.

This section uses two classical tools: Hilbert's Nullstellensatz and Gr\"{o}bner bases. We will recall both the complex and real versions of Hilbert's Nullstellensatz, and use the real version to prove the existence of solutions with coefficients in $\R \cap \overline{\Q}$, and we will use the complex version, alongside Gr\"obner basis, which makes explicit why a $0$-dimensional variety forces all solutions to be algebraic.

%
%
%

\subsection{Hilbert's Nullstellensatz viewpoint}

We start with the complex viewpoint. 
Let $k$ be a field.  Let $I=\langle f_1,\ldots,f_s\rangle$ be the ideal generated over $k[X_1,\ldots,X_n]$ by 
$f_1,\ldots,f_s\in k[X_1,\ldots,X_n]$. By definition, its variety is 
\[
\V_k(I) = \{ (a_1, \ldots, a_n) \in k^n,~ h(\bold{a}) = 0 \mbox{ for all }h \in I\}. 
\]

We recall a weak version of Hilbert's Nullstellensatz and a useful corollary.

\begin{theorem}(\textbf{Weak Nullstellensatz}, \cite[Chapter 4.1, Theorem 1]{cox2007})
\label{theorem:weak_nss_cox}
If $K$ is an algebraically closed field  and $J \subseteq K[X_1, \ldots, X_n]$ is an ideal such that $\V_K(J) = \varnothing$, then $J = K[X_1, \ldots, X_n]$.
\end{theorem}

\begin{corollary}
\label[corollary]{cor:weak_nss_cor_commonzero_or_certificatefor1}
    For $K$ an  algebraically closed field, a collection of polynomials $f_1, \ldots f_s$ in $K[X_1, \ldots, X_n]$ either has a common zero in $K^n$ or there exist $g_1,\ldots ,g_s$ in $K[X_1, \ldots, X_n]$ such that $\sum_{i = 1}^ s g_i f_i = 1$.
\end{corollary}
\begin{proof}
    Let $J = \langle f_1, \ldots, f_s \rangle$. Now either $\V_K(J) = \varnothing$ or $\V_K(J) \neq \varnothing$. If $\V_K(J) = \varnothing$, then by \Cref{theorem:weak_nss_cox}, this implies that $J = K[X_1, \ldots, X_n]$ hence $1 \in J$ and therefore there exists $g_1, \ldots, g_s \in K[X_1, \ldots, X_n]$ such that $\sum_{i = 1}^s g_if_i = 1$. If $\V_K(J) \neq \varnothing$, then $\V_K(J)$ contains at least one common solution to $f_1,\ldots,f_s$.
\end{proof}

\begin{theorem}
Let $k$ be a field, $\bar{k}$ be its algebraic closure and $K$ be an algebraically closed extension of $\bar{k}$. If $f_1,\ldots,f_s\in k[X_1,\ldots,X_n]$ have a common solution in $K^n$, then they have a common solution in $\bar{k}^n$.
\end{theorem}
\begin{proof}
We will apply \Cref{cor:weak_nss_cor_commonzero_or_certificatefor1} twice on the tower of fields $k \subseteq \overline{k} \subseteq K$. 

By assumption, there is a common solution in $K^n$. Now, suppose by contradiction that there is no common solution in $\bar{k}^n$. 
By \Cref{cor:weak_nss_cor_commonzero_or_certificatefor1} on $\overline{k}$, there exist $g_1, \ldots, g_s$ in $\overline{k}[X_1, \ldots, X_n]$ such that $\sum_{i = 1}^s g_if_i = 1$. However, $g_i \in \overline{k}[X_1, \ldots, X_n] \subseteq K[X_1, \ldots, X_n]$, so by applying \Cref{cor:weak_nss_cor_commonzero_or_certificatefor1} on $K$, it must be that there is no common zero in $K^n$ which contradicts our assumption.
    \end{proof}

We next look at the real viewpoint. 
Let again $I=\langle f_1,\ldots,f_s\rangle$ be the ideal generated by $f_1,\ldots,f_s\in R[X_1,\ldots,X_n]$ for a real closed $R$. 
By definition, its zero set is
\[
\V_R(I) =\{ \bold{a} =  (a_1,\ldots,a_n) \in R^n,~h(\bold{a})=0 \mbox{ for all }h\in I\}.
\]
For $U\subseteq R^n$, the vanishing ideal $\I_R(U)$ is by definition
\[
\I_R(U) = \{ f \in R[X_1,\ldots,X_n],~f(\bold{a})=0 \mbox{ for all }\bold{a} \in U \}. 
\]

\begin{theorem}(\textbf{Real Nullstellensatz}, \cite[Corollary 4.1.8]{Bochnak_Coste_Roy_1998}) If $R$ is a real closed field and $I \subseteq R[X_1,\ldots,X_n]$ is an ideal, then
$\I_R(\V_R(I))=\sqrt[R]{I}$, where 
\[
\sqrt[R]{I}=\{ h \in R[X_1,\ldots,X_n],~\exists m \in \mathbb{N}, ~\exists p_1,\ldots,p_t \in R[X_1,\ldots,X_n],~h^{2m}+p_1^2+\cdots+p_t^2 \in I \}.
\]    
\end{theorem}

\begin{corollary}
\label[corollary]{cor:real}
For $R$ a real closed field, a collection of polynomials $f_1,\ldots,f_s\in R[X_1,\ldots,X_n]$ either has a common zero in $R^n$ or there exist polynomials $g_1,\ldots,g_s$, $p_1,\ldots,p_t \in R[X_1,\ldots,X_n]$ such that
\[
\sum_{i=1}^sg_i f_i = 1 + \sum_{i=1}^tp_i^2.
\]
\end{corollary}
\begin{proof}
Set $I=\langle f_1,\ldots, f_s \rangle$. Now either $\V_R(I) = \varnothing$ or $\V_R(I) \neq \varnothing$. If $\V_R(I) = \varnothing$, then $\I_R(\V_R(I))=R[X_1,\ldots,X_n]$. Invoking the Real Nullstellensatz, we have that $1 \in \sqrt[R]{I}$ and there exist $p_1,\ldots,p_t$ such that $1^{2m}+p_1^2+\cdots+p_t^2 \in I$, hence by the definition of $I$, there exist $g_1,\ldots,g_s\in R[X_1,\ldots,X_n]$ such that
$\sum_{i=1}^s g_i f_i = 1+\sum_{i=1}^t p_i^2$.
\end{proof}

\begin{theorem}\label{theor:real}
Let $F\hookrightarrow R$ be an inclusion between two real closed fields. If $f_1,\ldots,f_s\in F[X_1,\ldots,X_n]$ have a common solution in $R^n$, then they have a common solution in $F^n$.
\end{theorem}
\begin{proof}
We will apply \Cref{cor:real} twice. 

By assumption, there is a common solution in $R^n$. Now, suppose by contradiction that there is no common solution in $F^n$. 
By \Cref{cor:real} on $F$, there exist $g_1, \ldots, g_s,p_1,\ldots,p_{t}$ in $F[X_1, \ldots, X_n]$ such that $\sum_{i = 1}^s g_if_i = 1 + \sum_{i=1}^tp_i^2$.  However, $g_i,p_i \in F[X_1, \ldots, X_n] \subseteq R[X_1, \ldots, X_n]$, so by applying \Cref{cor:real} on $R$, it must be that there is no common zero in $R^n$ which contradicts our assumption.
    \end{proof}

We can now explain why reported constructions of equiangular lines such as \cite{scott_grassl_2009} have algebraic coefficients.

\begin{theorem}\label{thm:existance_of_real_alg_lines_in_Cd}
If there exists a construction of $d^2$ equiangular lines in $\C^d$, there exists a construction of $d^2$ equiangular lines in $\Qbar^d$.  
\end{theorem}

\begin{proof}
All we need to do is represent a maximal equiangular set of lines $\{\bold{u}_1,\ldots,\bold{u}_{d^2}\}$ as a solution to a system of polynomial equations over $\Q$, which we have shown earlier in \Cref{eq:psic} is possible, namely:
\begin{equation*}
      p_{jl}(A_{1,1},\ldots,B_{d^2,d}) = \begin{cases}
           1, \quad j = l \\
           \frac{1}{d + 1} \quad j \neq l.
       \end{cases}
    \end{equation*}
We then apply \Cref{theor:real} with $n = 2d^3$, $F=\R \cap \overline{\Q}$ and $R=\R$, which are both real closed fields, and $\Q\subseteq F$. 
By assumption there exists a solution in $R^n$, so there is a solution in $F^n$, that is there are $A_{j,k},B_{j,k}$ real algebraic for $1 \leq j \leq d^2, 1 \leq k \leq d$.
Hence the coefficients $u_{jk} = A_{j,k}+iB_{j,k}$ of $\bold{u}_j$ are algebraic.
\end{proof}

In words, coefficients appearing in unit vectors forming a SIC-POVM are found in number fields. We observe that the complex version of Hilbert's Nullstellensatz is not enough for our purpose, since even if we knew that we had real solutions in $\R^{2d^3}$ to the polynomials among the complex ones in $\C^{2d^3}$, we could only guarantee that a subset of those are in $\overline{\Q}^{2d^3}$, and this subset may or not a priori contain real algebraic solutions, that is be in $(\R \cap \Qbar)^{2d^3}$.

\begin{theorem}\label{thm:wh}
If there exists a construction of a Weyl-Heisenberg fiducial vector in $\C^d$, there exists a construction of a Weyl-Heisenberg fiducial vector in $\Qbar^d$.  
\end{theorem}

\begin{proof}
We represent a maximal equiangular set of lines generated by a fiducial vector and the Weyl-Heisenberg group $\bold{v}=(v_0,\ldots,v_{d - 1})^T$ as a solution to a system of polynomial equations over a subfield of $\R$, which we have shown earlier in \Cref{eq:fid} is possible, namely:
\begin{equation*}
      p_{a,b}(A_0,\ldots,A_{d-1},B_0,\ldots,B_{d-1}) = \begin{cases}
        1, \quad a = b = 0 \\
        \frac{1}{d + 1}, \quad (a, b) \in \Z_d^2/\{(0, 0)\}
    \end{cases}
    \end{equation*}
which forms a system of polynomial equations with coefficients in $\Q(\cos(\tfrac{2\pi}{d}), \sin(\tfrac{2\pi}{d}))$.
We then apply \Cref{theor:real} with $n = 2d, F=\R \cap \overline{\Q}$ and $R=\R$, which are both real closed fields, and $\Q(\cos(\tfrac{2\pi}{d}), \sin(\tfrac{2\pi}{d}))\subseteq F$ since $\cos(\tfrac{2\pi}{d}), \sin(\tfrac{2\pi}{d})$ are real algebraic. 
Indeed, the root of unity $e^{\tfrac{2\pi i}{d}}$ is algebraic, and hence its conjugate is as well, resulting in both $e^{\tfrac{2\pi i}{d}}+e^{-\tfrac{2\pi i}{d}}= 2 \cos(\tfrac{2\pi}{d})$ and $e^{\tfrac{2\pi i}{d}}-e^{-\tfrac{2\pi i}{d}}= 2i\sin(\tfrac{2\pi}{d})$ being algebraic. Coefficients $v_{j}$ of $\bold{v}$ are of the form $A_{j}+iB_{j}$, for $A_{j},B_{j}$ algebraic, thus $v_{j}$ is algebraic.
\end{proof}

In the real case, we have
\begin{equation*}
  p_{j,l}(u_{1,1},\ldots,u_{N,d}) = \begin{cases}
       1, \quad j = l \\
       \alpha \quad j \neq l.
   \end{cases}
\end{equation*}

To decide where the polynomials $p_{j,l} - \alpha$ lives, we first need to establish where $\alpha$ lives. For our purpose, it is enough to know that $\alpha \in \R \cap \Qbar$ so that we may apply \Cref{theor:real}.

\begin{lemma}
\label[lemma]{lemma:alpha_in_R_union_Qbar}
    For a set of real equiangular lines, with common angle $\alpha$, it must be that $\alpha \in \R \cap \Qbar$.
\end{lemma}
\begin{proof}
    Let $\{\bold{u}_1, \ldots, \bold{u}_N\}$ be a set of equiangular lines where $d < N \leq \binom{d + 1}{2}$ and $\innerproduct{\bold{u}_j}{\bold{u}_l} = \pm \alpha$. Now construct the Gram matrix 
    $$
    G = \begin{pmatrix}
    \bold{u}_1^T \\
    \vdots \\
    \bold{u}_N^T
    \end{pmatrix}
    \begin{pmatrix}
    \bold{u}_1 & \ldots & \bold{u}_N
    \end{pmatrix}
    $$ 
    using this set. We know that $\Rank(G) = \dim(\Span\{\bold{u}_1, \ldots, \bold{u}_N\}) \leq d < N$, hence, it must be that $\det(G) = 0$. The determinant is a polynomial in $\alpha$ over $\Z$, not necessarily monic, thus $\alpha \in \Qbar$. Furthermore it is real and hence $\alpha \in \R \cap \Qbar$.
\end{proof}

\begin{theorem}
\label{theorem:existance_of_real_alg_sic_in_Rd}
    If there exist constructions of $n$ equiangular lines in $\R^d$, there exist constructions of $n$ equiangular lines in $(\overline{\Q} \cap \R)^d$.
\end{theorem}
\begin{proof}
    We represent a maximal equiangular set of lines $\{\bold{u}_1,\ldots,\bold{u}_{N}\}$ as a solution to a system of polynomial equations over $\Q$, which we have shown in \Cref{eq:realsystem} is possible, namely:
    \begin{equation*}
          p_{j, l}(u_{1,1},\ldots,u_{N,d}) = \begin{cases}
               1, \quad j = l \\
               \alpha \quad j \neq l.
           \end{cases}
        \end{equation*}
    We then apply \Cref{theor:real} with $n = Nd, F= \R \cap \overline{\Q}$ and $R= \R$, noting that $p_{j, l} - \alpha \in F$. This tells us that there exists a solution where $u_{j, k}$ are real algebraic for $1 \leq j \leq N, 1 \leq k \leq d$ as desired.
\end{proof}

%
%
%

\subsection{Gr\"obner bases viewpoint}

Given two vectors $\alpha = (\alpha_1, \ldots, \alpha_n), \beta = (\beta_1, \ldots, \beta_n)$ of positive integers,
they are ordered with respect to the lexicographic ordering, denoted by $\alpha >_{lex} \beta$, if the leftmost entry of the vector difference $\alpha - \beta > 0$. It will be assumed that we are using lexicographic ordering unless otherwise stated. Using the lexicographic order, for $f = \sum_{\alpha} a_{\alpha}x^{\alpha}$ a nonzero polynomial in $k[X_1, \ldots, X_n]$ for $k$ a field, its multidegree is by definition
\[
multideg(f) = \max\{\alpha \in \Z^n_{\geq 0} : a_{\alpha} \neq 0\}
\]
while its leading term is the product of the leading coefficient and the leading monomial, that is 
\[
LT(f) = LC(f) \cdot LM(f)
\]
where $LC(f) = a_{multideg(f)} \in k$ and $LM(f) = x^{multideg(f)}$.

We recall that a finite subset $G = \{g_1, \ldots, g_t\}$ of an ideal $\{0\} \neq I \subseteq k[X_1, \ldots, X_n]$ is said to be a Gr\"{o}bner basis if 
    \begin{equation*}
        \langle LT(g_1), \ldots, LT(g_t) \rangle = \langle LT(I) \rangle,
    \end{equation*}        
where $LT(I)$ is the set of leading terms of nonzero elements of $I$, that is
    \begin{equation*}
        LT(I) = \{cx^{\alpha} : \exists f \in I / \{0\}, \text{ such that } LT(f) = cx^{\alpha}\}
    \end{equation*}
for $\{0\} \neq I \subseteq k[X_1, \ldots, X_n]$.

Buchberger's Algorithm is a classical method to compute a Gr\"obner basis.

\begin{theorem}
\label{theorem:buchberger_algorithm}(\textbf{Buchberger's Algorithm} \cite[Chapter 2.7, Theorem 2]{cox2007}). 
    Let $I = \langle f_1, \ldots, f_s \rangle \neq \{0\}$ be a polynomial ideal. Then a Gr\"{o}bner basis can be constructed in finitely many steps by the following algorithm:
    \begin{algorithmic}[1]
    \Require $F = (f_1, \ldots, f_s)$
    \State $G = F$
    \Repeat
        \State $G' = G$
        \For{each pair $\{p, q\}, p \neq q$ in $G'$}
            \State $r = \overline{S(p, q)}^{G'}$
            \If{$r \neq 0$}
                \State $G = G \cup \{r\}$
            \EndIf
        \EndFor
    \Until{$G = G'$}
    \State \Output $G$
    \end{algorithmic}
where  
\begin{equation*}
S(p, q) = \frac{x^{\gamma}}{LT(p)}\cdot p- \frac{x^{\gamma}}{LT(q)} \cdot q
\end{equation*}
is the $S$-polynomial of $p, q$ and $\gamma = (\gamma_1, \ldots, \gamma_n)$, $\gamma_i = max(\alpha_i, \beta_i)$ with 
$\alpha = multideg(f), \beta = multideg(g)$ 
and $\overline{f}^F$ is the remainder when $f$ is divided by $F = (f_1, \ldots, f_s)$.
\end{theorem}

The following fundamental theorem about Gr\"{o}bner basis is recalled for the sake of completeness.

\begin{theorem}
\label{theorem:elimination_theorem_grobner_basis}(\textbf{Elimination Theorem}) \cite[Chapter 3.1, Theorem 2]{cox2007}) 
    Let $I \subseteq k[X_1, \ldots, X_n]$ and let $G$ be a Gr\"{o}bner basis of $I$, then for all $1 \leq l \leq n$ the set
    \begin{equation*}
        G_l = G \cap k[X_{l + 1}, \ldots, X_n]
    \end{equation*}
    is a Gr\"{o}bner basis of the $l$-th elimination ideal $I_l= I \cap k[X_{l + 1}, \ldots, X_n] \in k[X_{l + 1}, \ldots, X_n]$.
\end{theorem}

Let $k$ be a field and $K$ be an algebraically closed extension of $k$. 
Let $I_k=\langle f_1,\ldots,f_s\rangle$ be the ideal generated over $k[X_1,\ldots,X_n]$ by 
$f_1,\ldots,f_s\in k[X_1,\ldots,X_n]$
and $I_K$ be the ideal generated over $K[X_1,\ldots,X_n]$ by 
$f_1,\ldots,f_s\in k[X_1,\ldots,X_n]$, that is, given  $f_1,\ldots,f_s\in k[X_1,\ldots,X_n]$
\begin{align}
\begin{split}
    \label{eq:Ik}
    I_k &= \{\sum g_i f_i : g_i \in k[X_1, \ldots, X_n]\}
\end{split} \\
\begin{split}
    \label{eq:IK}
    I_K &= \{\sum g_i f_i : g_i \in K[X_1, \ldots, X_n]\}.    
\end{split}
\end{align}

By definition, their respective varieties are given by
\begin{align*}
 \V_K(I_k) &=
 \{\bold{a} = (a_1, \ldots, a_n) \in K^n : h(\bold{a}) = 0 \mbox{ for all }h \in I_k\} \\
 \V_K(I_K) &=
 \{\bold{a} = (a_1, \ldots, a_n) \in K^n : h(\bold{a}) = 0 \mbox{ for all }h \in I_K\}. 
\end{align*}
 
We show that extending the base field of the ideal $I$ does not change the resulting variety.

\begin{lemma}
\label[lemma]{lemma:VkK}
    We have $\V_K(I_K) = \V_K(I_k)$.
\end{lemma}
\begin{proof}
    If $\bold{a} \in \V_K(I_k)$, this implies $f_i(\bold{a}) = 0$ for all $f_i$, $i=1,\ldots,s$, and for $h \in I_K$ we have $h = \sum g_i f_i$, which implies $h(\bold{a}) = 0$, that is $\V_K(I_k) \subseteq \V_K(I_K)$. The other direction is clear since $I_k \subseteq I_K$ implies $\V_K(I_k) \supseteq \V_K(I_K)$.
\end{proof}
Furthermore, extending the base field of the ideal $I$ does not change the Gr\"{o}bner basis of the elimination ideals.
\begin{proposition}
\label[proposition]{prop:elimination_ideal_for_tower_of_fields_are_equal}
    Let $k$ be a field and $K$ be a field extension of $k$. Let $I = \langle f_1, \ldots, f_s \rangle$ where $f_i \in k[X_1, \ldots, X_n]$ for all $i$, then the Gr\"{o}bner basis of the elimination ideal of $I$ in $k$ and in $K$ are equal, and it is $G \cap k[X_{l + 1}, \ldots, X_n]$, where $G$ is the Gr\"{o}bner basis of $I_k$, that is
    \begin{equation*}
        I_k \cap k[X_{l + 1}, \ldots, X_n] = I_K \cap K[X_{l + 1}, \ldots, X_n] = \langle G \cap k[X_{l + 1}, \ldots, X_n] \rangle.
    \end{equation*}
\end{proposition}
\begin{proof}
    Note that the Gr\"{o}bner basis $G$ which is computed by  \Cref{theorem:buchberger_algorithm} lives in the same ring $R$ as the polynomials $f_1,\ldots,f_s$, because the computations involve only operations which are closed in $R$. Since $f_i \in R=k[X_1, \ldots, X_n]$, we will get the same Gr\"obner basis  
    for both $I_k$ and $I_K$ since it will live in $k[X_1, \ldots, X_n]$. Now by \Cref{theorem:elimination_theorem_grobner_basis}, we have that the Gr\"{o}bner basis of $(I_l)_k$ is equal to $G \cap k[X_{l + 1}, \ldots, X_n]$ and that the Gr\"{o}bner basis of $(I_{l})_K$ is equal to $G \cap K[X_{l + 1}, \ldots, X_n]= G\cap k[X_{l + 1}, \ldots, X_n]$ since the coefficients of the polynomials in $G$ are in $k$, and hence the bases are equal.
\end{proof}

The next result seems to be considered as common knowledge by specialists of the topic, but in the absence of an explicit proof to refer to, and for making the result available to a more general audience, we provide two proofs.

\begin{theorem}(\textbf{Finite variety implies algebraic})
\label{theorem:finite_variety_implies_algebraic}
    Let $k$ be a field, $\bar{k}$ be its algebraic closure, and $K$ be an algebraically closed extension of $\overline{k}$. Let $I := I_K$ be the ideal generated over $K[X_1,\ldots,X_n]$ by 
$f_1,\ldots,f_s\in k[X_1,\ldots,X_n]$ . If $|\V_K(I)| < \infty$, then $\V_K(I) \subseteq \overline{k}^n$. Said otherwise, if $\bold{a} = (a_1, \ldots, a_n) \in \V_K(I)$ for $\V_K(I)$ finite, then $a_i$ is algebraic over $k$ for all $1 \leq i \leq n$.
\end{theorem}
\begin{proof}
    We prove the contrapositive, that is, if there is a solution in $K^n$ with at least one coordinate transcendental over $k$, then there are infinitely many solutions in $\V_K(I)$. Suppose $\bold{a} = (a_1, \ldots, a_n) \in \V_K(I)$, and assume $a_n$ is transcendental over $k$. The choice of the last coordinate is done without loss of generality, since if another coordinate were to be transcendental, we could change the order (it is currently the lex order) to have this coordinate in the last position. 
    Let $(I_{n - 1})_k := I_{k} \cap k[X_n]$ 
    be the $(n - 1)$-th elimination ideal for $I_k$ defined in \Cref{eq:Ik}. Now we prove the theorem in several steps:
    
    \textbf{Claim 1:} $(I_{n - 1})_k = (0)$.
    
    \textbf{Proof 1:} Let $f \in (I_{n - 1})_k$, we have that $f \in k[X_n]$, hence $f(\bold{a}) = f(a_n) = 0$, but since $a_n$ is transcendental over $k$, it must be that $f = 0$, that is $(I_{n - 1})_k = (0)$.

    Let
    $I_{n - 1} := (I_{n - 1})_K = I \cap K[X_n] $ 
    be the $(n - 1)$-th elimination ideal for $I_K$ defined in \Cref{eq:IK}. 

\textbf{Claim 2:} $I_{n - 1}  = (0)$.

\textbf{Proof 2:}
    \Cref{prop:elimination_ideal_for_tower_of_fields_are_equal} tells us that $(I_{n - 1})_k$ and $I_{n - 1}$ have the same Gr\"{o}bner basis, and since we know that $(I_{n - 1})_k = (0)$, we have the desired claim.

Now $I_{n - 1}  = (0)$ implies that there is no non-constant univariate polynomial for $X_n$. It is known that $I$ is $0$-dimensional if and only if it contains a non-constant univariate polynomial for each variable \cite[Lemma 6.50]{Becker1993}. Hence $I$ is not $0$-dimensional and $V_K(I)$ is not finite since $\V_K(I)$ is finite if and only if $I$ is a $0$-dimensional ideal.

\end{proof}

Below is an alternate proof of \Cref{theorem:finite_variety_implies_algebraic}, suggested by Markus Grassl:
\begin{proof}
    We have that $\V_K(I)$ is finite if and only if $I$ is a $0$-dimensional ideal. It is known that $I$ is $0$-dimensional if and only if it contains a non-constant univariate polynomial for each variable \cite[Lemma 6.50]{Becker1993}, that is, for $1 \leq i \leq n$, there is a $f_i(X_i) \in I \cap K[X_i]$. By picking a suitable ordering and then applying \Cref{prop:elimination_ideal_for_tower_of_fields_are_equal}, $I \cap K[X_i] = I \cap k[X_i]$. Hence $f_i(X_i) \in I \cap k[X_i]$. Then by the properties of a variety, $\V_K(I \cap k[X_i]) \subseteq \V_K(f_i)$,
    and since $f_i$ is univariate, we have $\V_K(f_i) = \V_{\overline{k}}(f_i)$. Therefore, if $\bold{a} = (a_1, \ldots, a_n) \in \V_K(I)$, then since $f_i \in I$, we have $f_i(a_i) = 0$ and hence $a_i \in \V_{\overline{k}}(f_i)$ which implies that $a_i$ is algebraic over $k$ for all $1 \leq i \leq n$.
\end{proof}

From \Cref{eq:psic} and \Cref{eq:fid}, we know that unit vectors and fiducial vectors representing $d^2$ equiangular lines in $\C^d$ can be represented as solutions to a system of polynomial equations with coefficients in $k$, for $k=\Q$, respectively $k=\Q(\cos(\tfrac{2\pi}{d}), \sin(\tfrac{2\pi}{d}))$. Let $I_k$ be the ideal generated over $k$ by the aforementioned system of polynomial equations.
Then, when $K = \C$, $\V_K(I_k)$ contains by definition common solutions to the system of polynomial equations. By \Cref{lemma:VkK},  we know that $\V_K(I_k)=\V_K(I_K)$, and by \Cref{theorem:finite_variety_implies_algebraic}, showing that $\V_K(I_K)$ is finite would imply that the solutions are algebraic. 

Let us compute an example for $d=4$ (numerical computations were done using both SAGEmath \cite{sagemath} and MAGMA \cite{magma}). We observe that the polynomial system defined by \Cref{eq:psic} involves $2d^3$ real variables while the polynomial system defined by \Cref{eq:fid} only contains $2d$ real variables. We choose the latter, corresponding to looking for a fiducial vector $\bold{v}$.    
\begin{example}
For example, when $a=b=0$, we have
\begin{align*}
p_{0,0} &= (\sum_{j=0}^3 C_j^{(0)}\alpha_{0j}-D_j^{(0)}\beta_{0j})^2 + (\sum_{j=0}^3 C_j^{(0)}\beta_{0j}+D_j^{(0)}\alpha_{0j})^2 \\ 
    &= (\sum_{j=0}^3 C_j^{(0)})^2 + (\sum_{j=0}^3 D_j^{(0)})^2, \quad \text{Since $\alpha_{0j} = 1, \beta_{0j} = 0$} \\
    &= (\sum_{j=0}^3 A_j^{2}+B_j^2)^2, \quad 
    \text{Since $D_j^{(0)} = 0$}\\ 
    &= (x_0^2+x_1^2+x_2^2+x_3^2+x_4^2+x_5^2+x_6^2+x_7^2)^2
\end{align*}
by setting
\begin{equation*}
    \bold{v} = \begin{pmatrix}
        x_0 \\
        \vdots \\
        x_3
    \end{pmatrix} + i\begin{pmatrix}
        x_4 \\
        \vdots \\
        x_7
    \end{pmatrix},
\end{equation*}
yielding the first polynomial equation $p_{00}=1$. 

Repeating the same computations for $(a,b)\neq (0,0)$, we finally get $9$ distinct polynomials
$p_{a, b} $ with coefficients in $ \Q(\cos(2\pi /4), \sin(2\pi /4)) = \Q$ satisfying
\begin{equation*}
   p_{a,b}(x_0,\ldots,x_{7}) = \begin{cases}
        1, \quad a = b = 0 \\
        \frac{1}{5}, \quad (a, b) \in \Z_4^2/\{(0, 0)\}.
    \end{cases}
\end{equation*}
These polynomials are:
\begin{align*}
    p_{0, 0} &= (x_{0}^{2} + x_{1}^{2} + x_{2}^{2} + x_{3}^{2} + x_{4}^{2} + x_{5}^{2} + x_{6}^{2} + x_{7}^{2})^2 \\
    p_{0, 1} = p_{0, 3} &= (x_{0}^{2} - x_{2}^{2} + x_{4}^{2} - x_{6}^{2})^{2} + (x_{1}^{2} - x_{3}^{2} + x_{5}^{2} - x_{7}^{2})^{2} \\
    p_{0, 2} &= (x_{0}^{2} - x_{1}^{2} + x_{2}^{2} - x_{3}^{2} + x_{4}^{2} - x_{5}^{2} + x_{6}^{2} - x_{7}^{2})^{2} \\
    p_{1, 0} = p_{3, 0} &= (x_{0} x_{1} + x_{1} x_{2} + x_{0} x_{3} + x_{2} x_{3} + x_{4} x_{5} + x_{5} x_{6} + x_{4} x_{7} + x_{6} x_{7})^{2} + \\
    &(x_{1} x_{4} - x_{3} x_{4} - x_{0} x_{5} + x_{2} x_{5} - x_{1} x_{6} + x_{3} x_{6} + x_{0} x_{7} - x_{2} x_{7})^{2} \\
    p_{1, 1} = p_{3, 3} &= (x_{0} x_{1} - x_{2} x_{3} + x_{3} x_{4} + x_{2} x_{5} + x_{4} x_{5} - x_{1} x_{6} - x_{0} x_{7} - x_{6} x_{7})^{2} + \\
    &(x_{1} x_{2} - x_{0} x_{3} - x_{1} x_{4} + x_{0} x_{5} + x_{3} x_{6} + x_{5} x_{6} - x_{2} x_{7} - x_{4} x_{7})^{2} 
\end{align*}
\begin{align*}    
    p_{1, 2} = p_{3, 2} &= (x_{0} x_{1} - x_{1} x_{2} - x_{0} x_{3} + x_{2} x_{3} + x_{4} x_{5} - x_{5} x_{6} - x_{4} x_{7} + x_{6} x_{7})^{2} + \\
    &(x_{1} x_{4} + x_{3} x_{4} - x_{0} x_{5} - x_{2} x_{5} + x_{1} x_{6} + x_{3} x_{6} - x_{0} x_{7} - x_{2} x_{7})^{2} \\
    p_{1, 3} = p_{3, 1} &= (x_{0} x_{1} - x_{2} x_{3} - x_{3} x_{4} - x_{2} x_{5} + x_{4} x_{5} + x_{1} x_{6} + x_{0} x_{7} - x_{6} x_{7})^{2} + \\
    &(x_{1} x_{2} - x_{0} x_{3} + x_{1} x_{4} - x_{0} x_{5} - x_{3} x_{6} + x_{5} x_{6} + x_{2} x_{7} - x_{4} x_{7})^{2}\\
    p_{2, 0} &= 4(x_{0} x_{2} + x_{1} x_{3} + x_{4} x_{6} + x_{5} x_{7})^{2}\\
    p_{2, 1} = p_{2, 3} &= 4(x_{2} x_{4} - x_{0} x_{6})^{2} + 4(x_{3} x_{5} - x_{1} x_{7})^{2} \\
    p_{2, 2} &= 4(x_{0} x_{2} - x_{1} x_{3} + x_{4} x_{6} - x_{5} x_{7})^{2}
\end{align*}

The inner product does not change
if we multiply $\bold{v}$ by a complex number of modulus one. This implies not only that, without loss of
generality, we may choose the first coordinate of $\bold{v}$ to be real, that is, $x_4 = 0$, but more importantly that if we do not fix this condition, we will have a degree of freedom that will yield infinitely many solutions. Hence, we add
\[
p_0 = x_4.
\]
 
We get 11 polynomial equations in total. Let $I$ be the ideal generated by these polynomials, that is 
\begin{equation}
\label{eq:ideal_gen_by_wh_zauner}
    I = \langle p_0, p_{0, 0} - 1, p_{a, b} - \frac{1}{5} \mbox{ for } 
    (a, b) \in \Z_4^2/\{(0, 0)\} 
    \rangle \subset \Q[x_0, \ldots, x_7].
\end{equation}

We compute a Gr\"{o}bner basis of $I$  using the lex order, and with a numerical solver, we get 1024 vectors $(x_0, \ldots, x_7)$, which means that the variety $V_{\C}(I)$ has dimension $0$. 

One of the solutions is (written with a finite number of digits)
\begin{equation*}
    \bold{v} = \begin{pmatrix}
        0.48571221409126403909152153177 \\ 0.60043369656069688700611847041 + 0.44989636690811813902417022753i \\
        0.20118858648686589293456281597i \\ -0.39924511007383099407155565445 + 0.035815847183145900067351304236i
    \end{pmatrix}.
\end{equation*}

Now by \Cref{theorem:finite_variety_implies_algebraic}, we know that this solution  (as well as all other solutions) is algebraic, and indeed 
\begin{equation*}
    \bold{v} = e^{-i\pi /8}(X \psi_{1a} + \omega_8 Y\psi_{1b})
\end{equation*}
where
\begin{equation*}
    X = \frac{\sqrt{3 - \frac{3}{\sqrt{5}}}}{2}, \quad Y = \frac{\sqrt{1 + \frac{3}{\sqrt{5}}}}{2}, \quad
    \psi_{1a} = \frac{1}{\sqrt{6}} \begin{pmatrix}
        \omega_8 + 1 \\
        i \\
        \omega_8 - 1 \\
        i
    \end{pmatrix}, \quad \psi_{1b} = \frac{1}{\sqrt{2}} \begin{pmatrix}
        0 \\
        1 \\
        0 \\
        -1
    \end{pmatrix}.
\end{equation*}

More precisely, let 
\begin{equation*}
    \bold{v}_k = e^{-i\pi /8}(X \psi_{1a} + \omega_8^k Y\psi_{1b}).
\end{equation*}
Then 4 solutions are $\bold{v}_k$ for $k = 1, 3, 5, 7$.
Unsurprisingly, $\bold{v}_k$ is a phase shift of a fiducial vector found by Zauner, given by $\psi_k = e^{i\pi /8}\bold{v}_k$ \cite[p.~62]{zauner_2011}. The details of the solutions are reported in \Cref{tab:d4allsolutions}.

\begin{table}[h]
\centering
\begin{tabular}{|p{2.8cm}|p{7cm}|c|}
\hline
$\#$ of solutions & As computed by SAGEmath/MAGMA. & 1024 \\
\hline
$\#$ of real solutions & As computed by SAGEmath/MAGMA  & 512 \\
\hline
$\#$ of real solutions up to sign   & To each real solution $(x_0,\ldots,x_7)$ corresponds the real solution $(-x_0,\ldots,-x_7)$.   & 256 \\
\hline
$\#$ of fiducial vectors  & To each fiducial vector $\bold{v}$ corresponds 16 fiducial vectors $\bold{V}^a\bold{U}^b\bold{v}$. Also $-\bold{v}\neq \bold{V}^a\bold{U}^b\bold{v}$. This number matches \cite{scott_grassl_2009}   &  16   \\
\hline
$\#$ of Zauner fiducial vectors & These are $\bold{v}_k$ (up to a phase shift), $k=1,3,5,7$ & 4 \\
\hline
\end{tabular}
\caption{\label{tab:d4allsolutions}
A detail of all solutions to the system of polynomial equations for Weyl-Heisenberg fiducial vectors for $d=4$.
}
\end{table}

\end{example}
%
%
%

\section{Consequences and Discussions}
\label{sec:consequences}

We observe that the results from  \Cref{thm:existance_of_real_alg_lines_in_Cd} and \Cref{thm:wh} provide a theoretical basis for studying constructions of SIC-POVMs and fiducial vectors, respectively, using number fields, regardless of the dimension $d$.
We derive some further consequences.

\subsection{About normalized overlaps}

We recall from \Cref{eq:condition_for_maximal_set_of_equiangular_lines} that
\begin{equation}
  |\innerproduct{\bold{u}_j}{\bold{u}_l}|^2 = \begin{cases}
           1, \quad j = l \\
           \frac{1}{d + 1} \quad j \neq l.
       \end{cases}
\end{equation}
This implies that
\begin{equation}\label{eq:normalizedoverlap}
\innerproduct{\bold{u}_j}{\bold{u}_l} = \begin{cases}
           e^{i\theta_{jj}}, \quad j = l \\
           \frac{e^{i\theta_{jl}}}{\sqrt{d + 1}} \quad j \neq l,
       \end{cases}
\end{equation}
where $e^{i\theta_{jl}}$ is called the phase factor in general, and in the particular case of Weyl-Heisenberg covariant SIC-POVMs, it is called the normalized overlap.

\begin{proposition}
\label[proposition]{prop:existance_of_alg_phase_factors}
\begin{enumerate}
\item 
If there exists a construction of $d^2$ equiangular lines in $\C^d$, there exists a construction whose phase factors $e^{i\theta_{jl}}$ are algebraic numbers.
\item 
If there exists a construction of a Weyl-Heisenberg vector in $\C^d$, there exists a construction whose corresponding normalized overlaps $e^{i\theta_{jl}}$ are algebraic numbers.
\end{enumerate}  
\end{proposition}

\begin{proof}
This is a direct consequence of  \Cref{thm:existance_of_real_alg_lines_in_Cd} and \Cref{thm:wh} and of the fact that in both cases, $\langle \bold{u}_j,\bold{u}_l\rangle$ is algebraic if the coefficients of $\bold{u}_j,\bold{u}_l$ are.
\end{proof}

This is a step towards answering \Cref{conj:overlap}, since for the normalized overlap $e^{i\theta_{jl}}$ to be an algebraic unit, it first needs to be an algebraic number.

The question of whether investigating $\theta_{jl}$ will help answer the conjecture is natural, however the following corollary shows that we will be looking for transcendental $\theta_{ij}$ which is not promising.
\begin{corollary}
\label[corollary]{cor:eitheta_alg_implies_theta_transcendental}
If $e^{i\theta_{jl}}$ is algebraic where $\theta_{jl}$ is not zero, then $\theta_{jl}$ is transcendental.   
\end{corollary}
\begin{proof}
This is a direct consequence of Lindemann–Weierstrass theorem which states that if $\alpha_1,\ldots,\alpha_n$ are algebraic numbers that are linearly independent over $\mathbb{Q}$, then $\mathbb{Q}(e^{\alpha_1},\ldots,e^{\alpha_n})$ has transcendence degree $n$ over $\Q$.     
\end{proof}    

We have shown that the existence of a Weyl-Heisenberg vector in $\C^d$ implies the existence of normalized overlaps which are algebraic numbers. To prove \Cref{conj:overlap}, the next two steps would be to show that they are algebraic integers, and finally that they are algebraic units. 
We prove the last step with the following proposition and corollary.
\begin{proposition}(\cite{keithconradrootsonacircle})
\label[proposition]{prop:irred_f_with_root_on_unit_circle_cond}
    Let $f(x) \in \Q[x]$ be irreducible with degree $n > 1$. If $f(x)$ has a root on the unit circle then $n$ is even and $x^nf(\frac{1}{x}) = f(x)$, that is $f$ is reciprocal, and hence its leading coefficient is equal to its constant term, $f(0)$.
\end{proposition}
\begin{proof}
    Let $\alpha$ be a root of $f(x)$ with $|\alpha| = 1$. Since $f$ has real coefficients, $\overline{\alpha} = \alpha^{-1} = \frac{1}{\alpha}$ is also a root of $f(x)$. Now, observe that the polynomial $g(x) = x^nf(\frac{1}{x})$ is also a degree $n$ polynomial in $\Q[x]$; furthermore it has leading coefficient $f(0)$. Now, $g(\alpha) = \alpha^n f(\frac{1}{\alpha}) = 0$, hence since $f$ is irreducible, $g$ must be a constant multiple of $f$, that is $g(x) = cf(x)$ for some $c \neq 0 \in \Q$. Evaluating at $x = 1$ gives $g(1) = f(1) = cf(1)$, and since $f(\pm 1) \neq 0$ (otherwise by irreducibility $f(x) = x \pm 1$, which contradicts $n > 1$), it must be that $c = 1$ and hence, $f(x) = x^nf(\frac{1}{x})$. Now, evaluating at $x = -1$ we have $f(-1) = (-1)^n f(-1)$ which implies that $n$ is even. Lastly, since $f(x) = x^nf(1/x)$ this implies that the coefficient of $x^n$ and $x^0$, which is $f(0)$, are equal.
\end{proof}
\begin{corollary}
\label[corollary]{cor:phase_alg_int_implies_alg_unit}
    If a phase factor is an algebraic integer, then it is an algebraic unit.    
\end{corollary}
\begin{proof}
    Let $\alpha = e^{i\theta}$ be a phase factor, hence $|\alpha| = 1$, and an algebraic integer. If $\alpha = \pm 1$ then it is clearly an algebraic unit, so assume not. Using \Cref{prop:irred_f_with_root_on_unit_circle_cond} with $f(x) = m_{\alpha, \Q}(x) \in \Z[x]$, the minimal polynomial of $\alpha$ over $\mathbb{Q}$,  immediately implies that $\frac{1}{\alpha}$ is a root of $x^nf(\frac{1}{x}) = f(x)$ which is a monic polynomial with integer coefficients and hence $\alpha^{-1} = \frac{1}{\alpha}$ is an algebraic integer, which makes $\alpha$ an algebraic unit.
\end{proof}

Hence, providing a partial conclusion to \Cref{conj:overlap} reduces to proving that normalized overlaps are algebraic integers given that they are algebraic numbers.

\subsection{About algebraic Weyl-Heisenberg SIC-POVMs}

In the literature, Weyl-Heisenberg SIC-POVMs are classified by their $PEC(d)$ orbits (See \cite{Appleby_2005_clifford_group} for an in-depth analysis of $PEC(d)$ orbits) following the notation of Scott and Grassl \cite[Table I]{scott_grassl_2009}. Using our \Cref{thm:existance_of_real_alg_lines_in_Cd} we can claim that if a SIC-POVM exists, then every fiducial vector is algebraic in at least one of the $PEC(d)$ orbits.

\begin{proposition}
\label[proposition]{prop:existance_of_alg_pecd_orbit}
    If there exists a construction of $d^2$ equiangular lines in $\C^d$, there exists a $PEC(d)$ orbit for which every fiducial vector is algebraic over $\Q$.
\end{proposition}
\begin{proof}
    Suppose a SIC-POVM exists in $\C^d$, then \Cref{thm:existance_of_real_alg_lines_in_Cd} ensures that an algebraic fiducial vector exists in $\C^d$, say $v$, and hence, since every element of $PEC(d)$ is a matrix with algebraic entries lying in a quadratic extension of a cyclotomic field, every vector in the $PEC(d)$ orbit of $v$, will also be algebraic over $\Q$.
\end{proof}
\begin{corollary}
\label[corollary]{cor:pecd_orbit_is_alg}
    For $d = 4, 5, 6, 10, 22$ every Weyl-Heisenberg fiducial vector is algebraic over $\Q$.
\end{corollary}
\begin{proof}
    In these dimensions, there is only one $PEC(d)$ orbit \cite[Table I]{scott_grassl_2009}.
\end{proof}

This is a step towards answering the long standing question on why all reported Weyl-Heisenberg fiducials are algebraic over $\Q$.

We also note that combining  \Cref{theorem:finite_variety_implies_algebraic} with \Cref{conj:wh} would answer the long standing question.

\begin{theorem}
    For $d > 3$, Weyl-Heisenberg covariant SIC-POVMs are always algebraic over $\Q$.
\end{theorem}
\begin{proof}
    Assuming \Cref{conj:wh}, we have that the polynomials defining a Weyl-Heisenberg SIC-POVM, which we have shown in \Cref{subsec:poly} are defined over $k = \Q(\cos(\tfrac{2\pi}{d}), \sin(\tfrac{2\pi}{d}))$, together with polynomial fixing the phase, define an algebraic variety of dimension zero, that is the variety is finite, and hence by \Cref{theorem:finite_variety_implies_algebraic} the variety is contained in $\overline{k} \subseteq \Qbar$ as desired.
\end{proof}

\subsection{About real equiangular lines}

\Cref{theorem:existance_of_real_alg_sic_in_Rd} implies the existence of algebraic equiangular lines for the real case. 
Let us look at some simple constructions of real equiangular lines (see \cite{tremain2008concreteconstructionsrealequiangular} for a variety of constructions). 
\begin{example}[$d = 2$]
    Take the set $\{\begin{pmatrix}
        1 \\
        0
    \end{pmatrix}, \begin{pmatrix}
        \frac{1}{2} \\
        \frac{\sqrt{3}}{2}
    \end{pmatrix}, \begin{pmatrix}
        \frac{1}{2} \\
        -\frac{\sqrt{3}}{2}
    \end{pmatrix}\}$. This is a set of equiangular lines where $\alpha = \frac{1}{2}$; one can think of these lines as the diagonals of a regular hexagon on a flat surface. 
    Coefficients live in the number field $\Q(\sqrt{3})$.
\end{example}

\begin{example}[$d = 3$]
    The rows of the following matrix form a maximal set of real equiangular lines in $\R^3$ with angle $\alpha = \frac{1}{\sqrt{5}}$:
    \begin{equation*}
        \begin{pmatrix}
            0 & \sqrt{\frac{5 - \sqrt{5}}{10}} & \sqrt{\frac{5 + \sqrt{5}}{10}} \\
            0 & -\sqrt{\frac{5 - \sqrt{5}}{10}} & \sqrt{\frac{5 + \sqrt{5}}{10}} \\
            \sqrt{\frac{5 - \sqrt{5}}{10}} & \sqrt{\frac{5 + \sqrt{5}}{10}} & 0 \\
            -\sqrt{\frac{5 - \sqrt{5}}{10}} & \sqrt{\frac{5 + \sqrt{5}}{10}} & 0 \\
            \sqrt{\frac{5 + \sqrt{5}}{10}} & 0 & \sqrt{\frac{5 - \sqrt{5}}{10}} \\
            \sqrt{\frac{5 + \sqrt{5}}{10}} & 0 & -\sqrt{\frac{5 - \sqrt{5}}{10}}
        \end{pmatrix}.
    \end{equation*}
The coefficients live in a quadratic extension of $\Q(\sqrt{5})$.
\end{example}

The two examples illustrate solutions with coefficients in a number field. In the real case, it can in fact be shown that every set of real equiangular lines must have real algebraic coefficients, up to an orthogonal transformation.

\begin{theorem}
    Given a set of real equiangular lines in $\R^d$, represented by unit vectors, $\{\bold{u}_1, \ldots, \bold{u}_N\}$, there exists an orthogonal matrix, $O$, such that $O\bold{u}_j \in (\R \cap \Qbar)^d$ for $1 \leq j \leq N$.
\end{theorem}
\begin{proof}
    Let $\{\bold{u}_1, \ldots, \bold{u}_N\}$ be a set of equiangular lines where $d < N \leq \binom{d + 1}{2}$ and $\innerproduct{\bold{u}_j}{\bold{u}_l} = \pm \alpha$. Now construct the Gram matrix
    $$
    G = \begin{pmatrix}
    \bold{u}_1^T \\
    \vdots \\
    \bold{u}_N^T
    \end{pmatrix}
    \begin{pmatrix}
    \bold{u}_1 & \ldots & \bold{u}_N
    \end{pmatrix}
    $$ 
    using this set.

    The characteristic polynomial of $G$ is given by $\chi_G(\lambda) = \det(G - \lambda I_N) \in \Q(\alpha)[\lambda]$, and since $\alpha \in \R \cap \Qbar$ by \Cref{lemma:alpha_in_R_union_Qbar}, we have that $\lambda \in \Qbar$. Since $G$ is a positive semi-definite matrix, the eigenvalues of $G$ are real so $\lambda \in \R \cap \Qbar$. For each eigenvalue, $G - \lambda I_N$ has real algebraic entries and thus, since Gaussian elimination only performs operations in the same field, the eigenspace $E_{\lambda}$ has a real algebraic basis; furthermore, by applying the Gram-Schmidt process to each $E_{\lambda}$, we get orthonormal basis vectors with real algebraic coefficients, since again throughout the process we may need to extend the field to include all the normalization factors, but these will all be real algebraic.

    Since $G$ is symmetric, it has a spectral decomposition, that is $G = Q^T \Lambda Q$ where $Q$ is a matrix whose columns are the normalized eigenvectors of $G$, which we have shown can be chosen to be real algebraic, and $\Lambda$ is a diagonal matrix where the diagonal entries are the corresponding eigenvalues of $G$. Now let $G = V^TV$ where $V = \sqrt{\Lambda}Q^T$ and we have that the columns of $V$ represent the vectors $\{\bold{u}_1, \ldots, \bold{u}_N\}$ up to an orthogonal mapping, and since we have shown that the entries of $\Lambda$, and hence $\sqrt{\Lambda},$ and $Q$ are real algebraic, the columns of $V$ must be as well.
\end{proof}

\subsection*{Acknowledgments}

The authors would like to thank Markus Grassl for kindly sharing some of his MAGMA code and for fruitful discussions, including the shorter proof of \Cref{theorem:finite_variety_implies_algebraic}.
    
\bibliographystyle{alpha}
\bibliography{references}

\end{document}